\documentclass[times,5p,sort&compress]{elsarticle}
\usepackage[utf8]{inputenc}
\pdfoutput=1
\usepackage[T1]{fontenc}
\usepackage[english]{babel}
\usepackage{amssymb}
\usepackage{amsmath}
\usepackage{paralist}
\usepackage{microtype}
\usepackage[text size=scriptsize]{todonotes}
\usepackage{mathtools}
\usepackage{hyphenat}
\usepackage{hyperref}
\usepackage{doi}
\usepackage{subcaption}
\usepackage[sort&compress,nameinlink,noabbrev,capitalize]{cleveref}

\newcommand{\poly}{\ensuremath{\text{poly}}}
\newcommand{\vft}{\ensuremath{V_{FT}}}
\newcommand{\vc}{\ensuremath{V_C}}
\newcommand{\chp}{tight}
\newcommand{\Chp}{Tight}
\newcommand{\hcpc}{HCPP(c)}
\newcommand{\hcpcl}{HCPP(c,l)}
\newcommand{\hcpl}{HCPP(l)}
\newcommand{\hcpp}{HCPP}
\newcommand{\algA}{\ensuremath{\mathcal A}}
\newcommand{\stRPPlong}{$s$-$t$-Rural Postman Path Problem}
\newcommand{\stRPP}{$s$-$t$-RPP}
\newcommand{\rpptime}{\tau}
\newcommand{\RPP}{RPP}
\newcommand{\stTSP}{$s$-$t$-TSP}
\newcommand{\Gpath}{layer path}
\newcommand{\prt}{\ensuremath{\mathcal{P}}}
\hypersetup{colorlinks,allcolors=blue}

\usepackage{tikz}
\usetikzlibrary{arrows,calc}
\usetikzlibrary{calc,shapes,decorations.shapes,snakes}
\usetikzlibrary{decorations.pathreplacing}
\pgfdeclarelayer{bg}    %
\pgfsetlayers{bg,main}  %

\tikzset{
    enull/.style = {densely dotted},
    edge/.style = {thick},
    matching/.style = {ultra thick},
    tour/.style = {->, thick},
    subset/.style = {line width=0.6cm,rounded corners,draw=white!85!black,fill=white!85!black}
}

\usepackage{amsthm}

\theoremstyle{definition}
\newtheorem{theorem}{Theorem}[section]
\newtheorem{observation}[theorem]{Observation}
\newtheorem{corollary}[theorem]{Corollary}
\newtheorem{construction}[theorem]{Construction}
\newtheorem{lemma}[theorem]{Lemma}
\newtheorem{problem}[theorem]{Problem}
\newtheorem{definition}[theorem]{Definition}
\newtheorem{proposition}[theorem]{Proposition}

\newtheoremstyle{iostuff}%
{0pt}%
{0pt}%
{\hangindent=\parindent}%
{}%
{\itshape}%
{:}%
{.5em}%
{}%

\theoremstyle{iostuff}
\newtheorem*{probinstance}{Input}
\newtheorem*{probfind}{Find}
\newtheorem*{probquest}{Question}

\crefname{construction}{Construction}{Constructions}
\crefname{paragraph}{Section}{Sections}
\crefname{lemma}{Lemma}{Lemmas}
\Crefname{lemma}{Lem.}{Lem.}
\crefname{theorem}{Theorem}{Theorems}
\Crefname{theorem}{Thm.}{Thm.}
\crefname{proposition}{Proposition}{Propositions}
\Crefname{proposition}{Prop.}{Props.}
\crefname{remark}{Remark}{Remarks}
\Crefname{remark}{Rem.}{Rem.}
\crefname{prop}{Property}{Properties}
\crefname{problem}{Problem}{Problems}
\crefname{observation}{Observation}{Observations}
\crefname{corollary}{Corollary}{Corollaries}
\Crefname{corollary}{Cor.}{Cors.}
\crefname{line}{line}{lines}
\crefname{section}{Section}{Sections}
\Crefname{section}{Sec.}{Secs.}

\newcommand{\ord}{\prec}
\newcommand{\N}{\mathbb N}

\newcommand{\w}{\omega}
\newcommand{\wG}{\w_{\Gamma}}
\newcommand{\wtG}{\widetilde \w_{\Gamma}}
\newcommand{\mul}[1]{\ensuremath{\mu_{#1}}}
\newcommand{\sg}[2]{\ensuremath{#1\langle #2\rangle}}
\begin{document}
\begin{frontmatter}
  \title{The Hierarchical Chinese
    Postman Problem:
    the slightest disorder makes it hard,
    \\
    yet
    disconnectedness is manageable} 
  \date{}
  \author[mmf]{Vsevolod A.\ Afanasev}
  \ead{v.afanasev3@g.nsu.ru}
  \author[mmf]{René van Bevern\corref{cor}}
  \ead{rvb@nsu.ru}
  \cortext[cor]{Correspondence to: Deptartment of Mechanics and Mathematics, Novosibirsk State University, ul.\ Pirogova 1, Novosibirsk, 630090, Russian Federation, \texttt{rvb@nsu.ru}}

  \author[mmf,im]{Oxana Yu.\ Tsidulko}
  \ead{o.tsidulko@nsu.ru}
  \address[mmf]{Department of Mechanics and Mathematics,
    Novosibirsk State University,
    Novosibirsk, Russian Federation}
  \address[im]{Sobolev Institute of Mathematics of the Siberian Branch of the
    Russian Academy of Sciences,
    Novosibirsk, Russian Federation}
  
  \begin{abstract}
    The Hierarchical Chinese Postman Problem is
    finding a shortest traversal of all edges of a graph
    respecting precedence constraints
    given by a partial order on classes of edges.
    We show that the special case
    with connected classes is NP-hard
    even on orders
    decomposable into a chain and an incomparable class.
    For the case with linearly ordered
    (possibly disconnected) classes,
    we get 5/3\hyp approximations and
    fixed\hyp parameter algorithms
    by transferring results
    from the Rural Postman Problem.
  \end{abstract}
\end{frontmatter}

\thispagestyle{empty}

\section{Introduction}
\label{intro}
\noindent
The following NP\hyp hard arc routing problem
arises in snow plowing,
garbage collection,
flame and laser cutting \citep{DST87,Lap15}.

\begin{problem}[Hierarchical Chinese Postman Problem, \hcpp{}]
  \label[problem]{prob:HCPP}
  \begin{probinstance}
    An undirected graph~$G=(V,E)$,
    edge weights~$\w\colon E\to \N$,
    a partition~$\prt$ of~$E$ into $k$~\emph{classes},
    a partial order~$\ord$ on~$\prt$.
  \end{probinstance}
  \begin{probfind}
    A least\hyp weight closed walk
    traversing each edge in~$E$ at least once
    such that each edge~$e$ in a class~$E'$
    is traversed only after all edges in all classes~$E''\ord E'$
    are traversed.
  \end{probfind}
\end{problem}

\noindent
The case $k=1$
is the \emph{Chinese Postman Problem (CPP)},
which
reduces to a %
minimum\hyp weight
perfect matching problem \citep{Chr73,EJ73,Ser74,BS20b}.
We study the
following special cases of \hcpp:

\medskip
\begin{compactitem}[\hcpcl:]
\item[\hcpl:]  the order $\ord$ is linear,
\item[\hcpc:] each edge class induces a connected subgraph,
\item[\hcpcl:] both of the above restrictions.
\end{compactitem}

\medskip\noindent
\hcpl{} and \hcpcl{}
can also be understood
as variants of the Travelling Salesman Problem (TSP)
in temporal graphs \citep{MS16b},
with the difference
that it is required 
to explore all edges instead of all vertices
and that edges never disappear from the graph.
\hcpcl{} is polynomial\hyp time solvable
\citep{DST87,GI00,KV06}.
This naturally raises two questions
about \hcpc{} and \hcpl{}:

\begin{inparaenum}[(a)]
\item\label{qwidth} Is \hcpc{} effectively solvable
  on other order types,
  like several scheduling
  problems on, for example,
  tree orders~\cite{Hu61},
  interval orders~\cite{PY79},
  and
  bounded\hyp width orders~\citep{Ser00,BBB+16}?

\item\label{qtype} Is \hcpl{}
  effectively solvable
  when the number of connected components
  in each class is sufficiently small?
  If the number of connected components
  is unbounded,
  \hcpl{} is NP\hyp hard already
  for~$k=2$ \citep{BNSW15}.  
\end{inparaenum}

\paragraph{Our contributions}
In \cref{sec:NPhard},
we show that \hcpc{} is NP\hyp hard
even %
on partial orders
that are decomposable into a linear order and
a class that is incomparable to all others classes,
thus negatively answering~\eqref{qwidth}
for all order types mentioned there.
The remaining sections
are dedicated to question~\eqref{qtype}.

In \cref{sec:RPP},
we revisit a construction
that reduces
\hcpcl{}
to the \stRPPlong{} (\stRPP, \cref{prob:strpp})~\citep{DST87}.
We show that,
when applied to \hcpl{},
the construction
transfers
performance guarantees
of approximation and randomized algorithms
from \stRPP{} to \hcpl{}.

In \cref{sec:apx},
we show a $5/3$\hyp approximation algorithm
for \hcpl{}.
This contrasts TSP in temporal graphs,
which is not better than 2\hyp approximable
unless P${}={}$NP \citep{MS16b}.
To get $5/3$\hyp approximations for \hcpl{},
we use
the construction from \cref{sec:RPP}
and show a $5/3$\hyp approximation algorithm for \stRPP{}
analogously to that for $s$-$t$-TSP \citep{Hoo91}.
Any better approximation factor for \stRPP{}
will directly carry over to \hcpl{}.

\looseness=-1
In \cref{sec:fpt},
we use the construction
from \cref{sec:RPP}
and known algorithms for the Rural Postman Problem \citep{EGL95,BNSW15,GWY17}
to show
that \hcpl{} is solvable in polynomial\hyp time
if each class induces
a constant number~$c$ of connected components.
When the edge weights are polynomially bounded,
one can even obtain randomized fixed\hyp parameter
algorithms with respect to~$c$. %

\section{Preliminaries}
\noindent
By \(\N\),
we denote the set of natural numbers,
including zero.
For two multisets~$A$ and~$B$,
$A\uplus B$ is the multiset
obtained
by adding the multiplicities
of elements in~$A$ and~$B$.
By $A\setminus B$ we denote the multiset
obtained by subtracting
the multiplicities of elements in~$B$
from the multiplicities of elements in~$A$.
Finally,
given some weight function~\(\w\colon A\to\N\),
the \emph{weight} of a multiset~$A$
is $\w(A):=\sum_{e\in A}\mathbb \nu(e)\w(e)$,
where $\mathbb \nu(e)$~is
the multiplicity of~$e$ in~$A$.

We mostly consider simple
undirected graphs~$G=(V,E)$
with a set~$V(G):=V$ of \emph{vertices}
and
a set~$E(G):=E\subseteq\{\{u,v\}\mid u,v\in V, u\ne v\}$
of \emph{edges}.
Unless stated otherwise,
$n$~denotes the number
of vertices and $m$~denotes the number of edges.
Within proofs,
there may occur \emph{multigraphs},
where $E$~is a \emph{multiset},
and directed graphs~$G=(V,A)$
with a set of \emph{arcs}~$A\subseteq V^2$.
The \emph{degree} of a vertex
in an undirected multigraph
is its number of incident edges.
We call a vertex \emph{balanced} if it has even degree.
We call a graph \emph{balanced}
if all its vertices are balanced.
For a multiset~$R$ of edges,
we denote by~$V(R)$ the set of their incident vertices.
For a multiset~$R$ of edges of~$G$,
$\sg G R:=(V(R),R)$
is the (multi)graph
\emph{induced by the edges in~$R$}.

A \emph{walk from~$v_0$ to~$v_\ell$}
in a graph~$G$ is a
sequence~$w=(v_0,e_1,\allowbreak v_1,\allowbreak e_2,v_2,\dots,\allowbreak
e_\ell,v_\ell)$
such that $e_i=\{v_{i-1},v_{i}\}$
(if $G$~is undirected)
or $e_i=(v_{i-1},v_i)$
(if $G$~is directed)
for each~$i\in\{1,\dots,\ell\}$.
When there is no ambiguity
(like in simple graphs),
we will also specify walks simply as a list of vertices.
If $v_0=v_\ell$,
then we call~$w$ a \emph{closed walk}.
If all vertices on~$w$ are pairwise distinct,
then $w$~is a \emph{path}.
If only its first and last vertex coincide,
then $w$~is a \emph{cycle}.
A \emph{subwalk~$w'$} of~$w$ is
any subsequence~$w'$ of~$w$ that is itself a walk.
By $E(w)$,
we denote the multiset of edges on~$w$,
that is,
each edge appears on~$w$ and in~$E(w)$ equally often.
The \emph{weight} of walk~$w$
is $\w(w):=\sum_{i=1}^\ell\w(e_\ell)$.
For a walk~$w$,
we denote $\sg{G}{w}:=\sg{G}{E(w)}$.
Note that $\sg{G}{R}$ and~$\sg{G}{w}$
do not contain isolated vertices yet
might contain
edges with a higher multiplicity than~$G$ and,
therefore,
are not necessarily sub(multi)graphs of~$G$.
An \emph{Euler walk for~$G$}
is a walk that traverses
each edge or arc of~$G$
exactly as often as it is present in~$G$.
An \emph{Euler tour} is a closed Euler walk.
A graph is \emph{Eulerian}
if it allows for an Euler tour.
A connected undirected graph is Eulerian
if and only if all its
vertices are balanced.

For any $\alpha\geq 1$,
an \emph{$\alpha$\hyp approximate solution}
for a minimization problem
is a feasible solution whose weight
does not exceed the weight of an optimal solution
by more than a factor of~$\alpha$,
called the \emph{approximation factor} \citep{GJ79}.
The Exponential Time Hypothesis (ETH)
is that 3-SAT (\cref{prob:3sat})
with $n$~variables
is not solvable in $2^{o(n)}$~time
\citep{IPZ01}.

\section{NP-hardness for \hcpc{} with one incomparable class}\label{sec:NPhard}

\noindent
\hcpcl{} is polynomial\hyp time solvable
\citep{DST87,GI00,KV06}.
We prove that adding a single incomparable class
makes the problem NP\hyp hard.

\begin{theorem}
  \label{thm:hard}
  Even on edges with weight one
  and orders that are
  decomposable into a linear order and
  a class that is incomparable to all others classes,
  \begin{compactenum}[(i)]
  \item \hcpc{} is NP-hard and
  \item not solvable in $2^{o(n+m+k)}$~time unless ETH fails.
  \end{compactenum}
\end{theorem}

\noindent
To prove \cref{thm:hard},
we use a polynomial\hyp time
many\hyp one reduction from 3-SAT,
which is
NP-hard \citep{Kar72}
and, unless the ETH fails,
is not solvable in $2^{o(n+m)}$~time \citep{IPZ01}.
\begin{problem}[3-SAT]\label[problem]{prob:3sat}
  \begin{probinstance}
    A formula~$\varphi$ in conjunctive normal form
    with $n$~variables and $m$~clauses,
    each containing at most three literals.
  \end{probinstance}
  \begin{probquest}
    Is there a assignment to the variables
    satisfying~$\varphi$?
  \end{probquest}
\end{problem}
\noindent
The reduction is carried out by the following construction,
which is illustrated in \cref{fig:constr}.

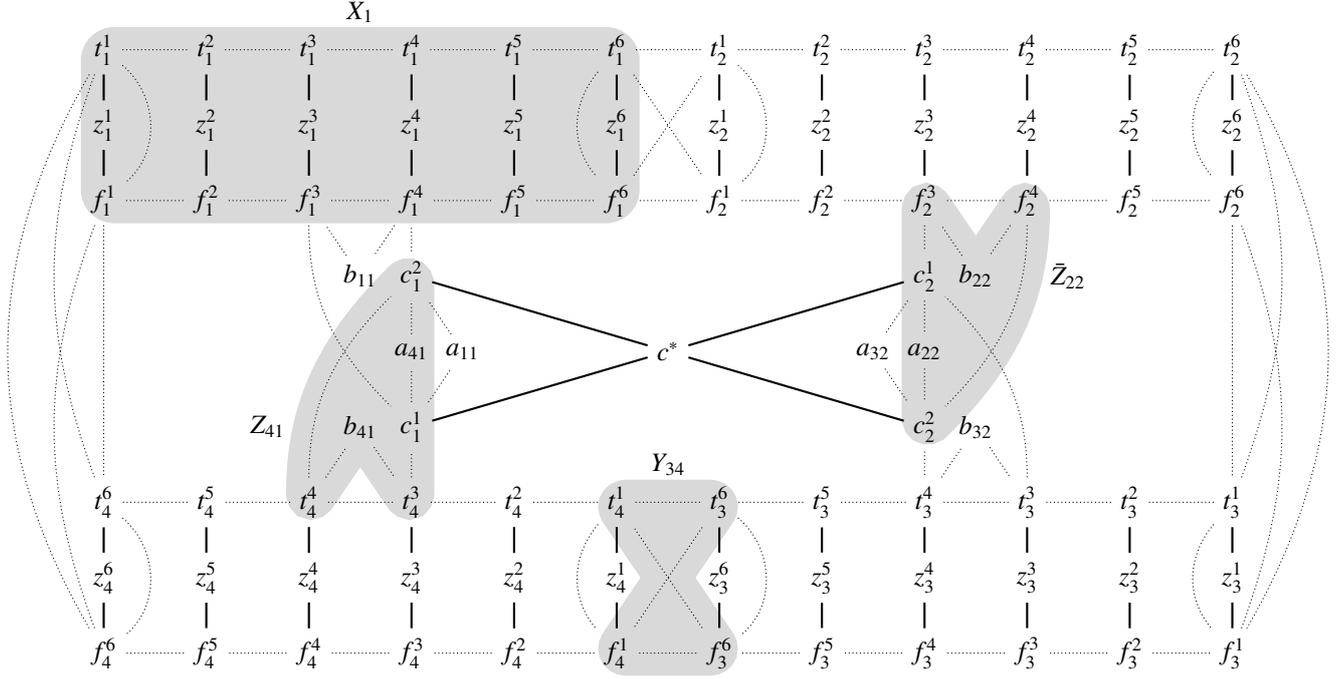
\begin{figure*}
  \centering
  \begin{tikzpicture}[x=1.35cm,y=1cm]
    \foreach \i in {1,...,6}
    {
      \node (t1\i) at (\i,0) {$t_1^{\i}$};
      \node (z1\i) at (\i,-1) {$z_1^{\i}$};
      \node (f1\i) at (\i,-2) {$f_1^{\i}$};
    }
    \foreach \i in {1,...,6}
    {
      \node (t2\i) at (6+\i,0) {$t_2^{\i}$};
      \node (z2\i) at (6+\i,-1) {$z_2^{\i}$};
      \node (f2\i) at (6+\i,-2) {$f_2^{\i}$};
    }
    \foreach \i in {1,...,6}
    {
      \node (t3\i) at (13-\i,-6) {$t_3^{\i}$};
      \node (z3\i) at (13-\i,-7) {$z_3^{\i}$};
      \node (f3\i) at (13-\i,-8) {$f_3^{\i}$};
    }
    \foreach \i in {1,...,6}
    {
      \node (t4\i) at (7-\i,-6) {$t_4^{\i}$};
      \node (z4\i) at (7-\i,-7) {$z_4^{\i}$};
      \node (f4\i) at (7-\i,-8) {$f_4^{\i}$};
    }
    \foreach \j in {1,...,4} {
      \foreach \i in {1,...,6} {
        \draw[edge] (t\j\i)--(z\j\i)--(f\j\i);   
      }
    }
    
    \node (a11) at (4.5,-4) {$a_{11}$};
    \node (b11) at (3.5,-3) {$b_{11}$};
    \node (a41) at (4,-4) {$a_{41}$};
    \node (b41) at (3.5,-5) {$b_{41}$};
    \node (a22) at (9,-4) {$a_{22}$};
    \node (b22) at (9.5,-3) {$b_{22}$};
    \node (a32) at (8.5,-4) {$a_{32}$};
    \node (b32) at (9.5,-5) {$b_{32}$};

    \node (c*) at (6.5,-4) {$c^*$};
    \node (c11) at (4,-5) {$c_1^1$};
    \node (c12) at (4,-3) {$c_1^2$};
    \node (c21) at (9,-3) {$c_2^1$};
    \node (c22) at (9,-5) {$c_2^2$};
    
    \draw[enull] (f13) to[out=-90] (c11);
    \draw[enull] (c11)--(a11)--(c12)--(f14)--(b11)--(f13);
    \draw[enull] (f23)--(c21)--(a22)--(c22);
    \draw[enull] (f24) to[out=-90, in=45](c22);
    \draw[enull] (f24) -- (b22)--(f23);
    \draw[enull] (t43)--(c11)--(a41)--(c12);
    \draw[enull] (t44) to[out=90,in=180+45] (c12);
    \draw[enull] (t44) -- (b41)--(t43);
    \draw[enull] (t33) to[out=90,in=-45] (c21);
    \draw[enull] (c21)--(a32)--(c22)--(t34)--(b32)--(t33);

    \foreach \j in {1,...,2} {
      \foreach \i in {1,...,2} {
        \draw[edge] (c*)--(c\i\j);
      }
    }
    \foreach \i in {1,...,4} {
      \draw[enull] (t\i1)--(t\i2)--(t\i3)--(t\i4)--(t\i5)--(t\i6);
      \draw[enull] (f\i1)--(f\i2)--(f\i3)--(f\i4)--(f\i5)--(f\i6);
    }
    \draw[enull] (t36) to[out=-40,in=40] (f36);
    \draw[enull] (t46) to[out=-40,in=40] (f46);
    \draw[enull] (t31) to[out=-130,in=130] (f31);
    \draw[enull] (t41) to[out=-130,in=130] (f41);
    \draw[enull] (t11) to[out=-40,in=40] (f11);
    \draw[enull] (t21) to[out=-40,in=40] (f21);
    \draw[enull] (t16) to[out=-130,in=130] (f16);
    \draw[enull] (t26) to[out=-130,in=130] (f26);
    \draw[enull] (t16)--(f21)--(f16)--(t21)--(t16);
    \draw[enull] (t36)--(f41)--(f36)--(t41)--(t36);
    
    \draw[enull] (t26) to[out=-60,in=60] (f31);
    \draw[enull] (f31) to[out=70,in=-70] (f26);
    \draw[enull] (f26)--(t31);
    \draw[enull] (t31) to[out=70,in=-70] (t26);
    
    \draw[enull] (t11) to[out=180+60,in=180-60] (f46);
    \draw[enull] (f46) to[out=180-70,in=180+70] (f11);
    \draw[enull] (f11)--(t46);
    \draw[enull] (t46) to[out=180-70,in=180+70] (t11);

    \begin{pgfonlayer}{bg}
      \path[subset] (t11.center) -- (f11.center) -- (f16.center) -- (6,0) -- cycle;
      \path[subset] (f41.center) -- (f36.center) -- (t41.center) -- (t36.center) -- cycle;
      \path[subset] (a22.center) -- (c21.center) -- (f23.center) -- (b22.center) -- (f24.center) to[out=-90,in=45] (c22.center) -- cycle;
      \path[subset] (a41.center) -- (c11.center) -- (t43.center) -- (b41.center) -- (t44.center) to[out=90,in=180+45] (c12.center) -- cycle;
      \node (X11) at (3.5,0.5) {$X_{1}$};
      \node (Y34) at (6.5,-5.5) {$Y_{34}$};
      \node (Z22) at (10.4,-3) {$\bar Z_{22}$};
      \node (Z41) at (2.6,-5) {$Z_{41}$};
    \end{pgfonlayer}
    \end{tikzpicture}
    \caption{Illustration of \cref{constr:3satred}.
      The graph is generated from the formula~$\varphi=(\bar x_1\vee x_4)\wedge (\bar x_2\vee x_3)$, that is, $C_1=(\bar x_1\vee x_4)$ and~$C_2=(\bar x_2\vee x_3)$.
    The dotted edges form the edge set~$E_0$.
    The solid edges form the paths~$P_i^\ell$ and $(c_j^1,c^*,c_j^2)$.
    The gray areas illustrate the types of cycles introduced in the construction:
    they consist of the dashed edges enclosed in the gray areas.
    Note that~$E_0$ (the dotted edges)
    forms an Eulerian subgraph: it is the union of cycles
    and thus each vertex has an even number of incident edges in~$E_0$.
    Thus, each vertex of the form~$t_i^\ell$, $f_i^\ell$, $c_j^1$ and~$c_j^2$
  is imbalanced and they are the only imbalanced vertices.}
\label{fig:constr}
\end{figure*}
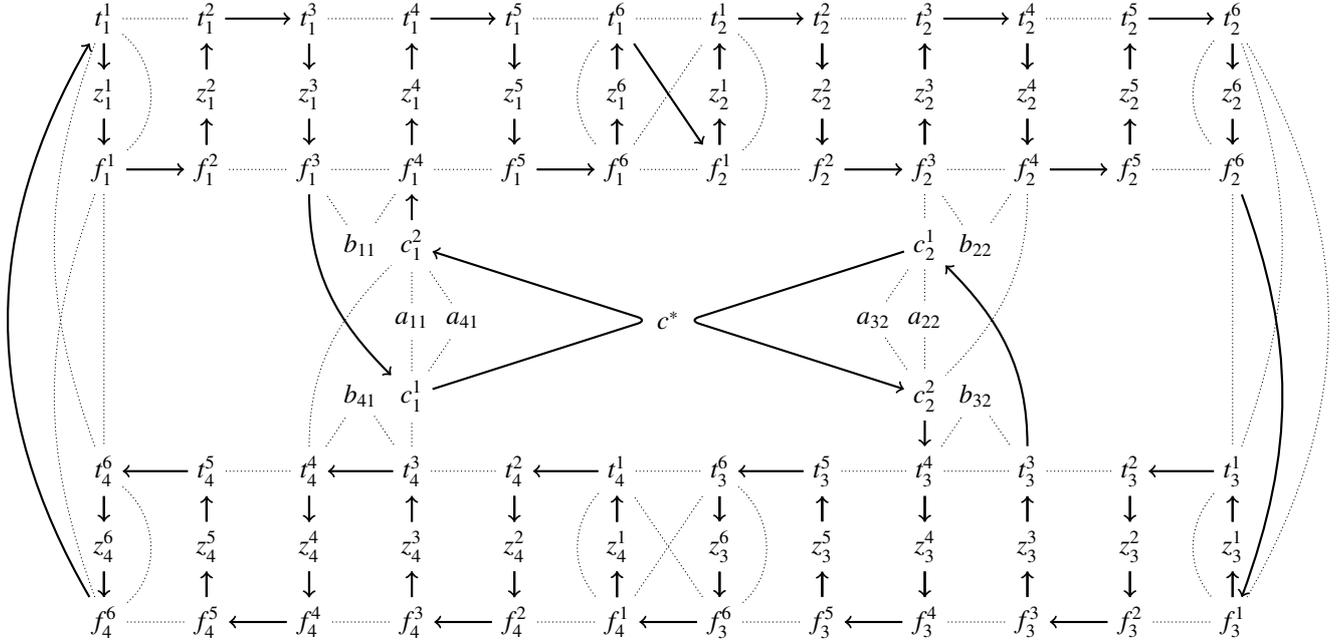
\begin{figure*}
  \centering
  \begin{tikzpicture}[x=1.35cm,y=1cm]
    \foreach \i in {1,...,6}
    {
      \node (t1\i) at (\i,0) {$t_1^{\i}$};
      \node (z1\i) at (\i,-1) {$z_1^{\i}$};
      \node (f1\i) at (\i,-2) {$f_1^{\i}$};
    }
    \foreach \i in {1,...,6}
    {
      \node (t2\i) at (6+\i,0) {$t_2^{\i}$};
      \node (z2\i) at (6+\i,-1) {$z_2^{\i}$};
      \node (f2\i) at (6+\i,-2) {$f_2^{\i}$};
    }
    \foreach \i in {1,...,6}
    {
      \node (t3\i) at (13-\i,-6) {$t_3^{\i}$};
      \node (z3\i) at (13-\i,-7) {$z_3^{\i}$};
      \node (f3\i) at (13-\i,-8) {$f_3^{\i}$};
    }
    \foreach \i in {1,...,6}
    {
      \node (t4\i) at (7-\i,-6) {$t_4^{\i}$};
      \node (z4\i) at (7-\i,-7) {$z_4^{\i}$};
      \node (f4\i) at (7-\i,-8) {$f_4^{\i}$};
    }
    \foreach \j in {1,...,4} {
      \foreach \i in {1,...,6} {
        \draw[edge] (t\j\i)--(z\j\i)--(f\j\i);   
      }
    }
    
    \node (a11) at (4,-4) {$a_{11}$};
    \node (b11) at (3.5,-3) {$b_{11}$};
    \node (a41) at (4.5,-4) {$a_{41}$};
    \node (b41) at (3.5,-5) {$b_{41}$};
    \node (a22) at (9,-4) {$a_{22}$};
    \node (b22) at (9.5,-3) {$b_{22}$};
    \node (a32) at (8.5,-4) {$a_{32}$};
    \node (b32) at (9.5,-5) {$b_{32}$};

    \node (c*) at (6.5,-4) {$c^*$};
    \node (c11) at (4,-5) {$c_1^1$};
    \node (c12) at (4,-3) {$c_1^2$};
    \node (c21) at (9,-3) {$c_2^1$};
    \node (c22) at (9,-5) {$c_2^2$};
    
    \draw[tour ] (f13) to[out=-90] (c11);
    \draw[tour, rounded corners] (c11) --(c*.west) -- (c12);
    \draw[tour, rounded corners] (c21) --(c*.east) -- (c22);
    \draw[tour] (c12) -- (f14);
    \draw[enull] (c11)--(a11)--(c12)--(f14)--(b11)--(f13);
    \draw[enull] (f23)--(c21)--(a22)--(c22);
    \draw[enull] (f24) to[out=-90, in=45](c22);
    \draw[enull] (f24) -- (b22)--(f23);
    \draw[enull] (t43)--(c11)--(a41)--(c12);
    \draw[enull] (t44) to[out=90,in=180+45] (c12);
    \draw[enull] (t44) -- (b41)--(t43);
    \draw[tour] (t33) to[out=90,in=-45] (c21);
    \draw[enull] (c21)--(a32)--(c22);
    \draw[tour] (c22)--(t34);
    \draw[enull] (t34)--(b32)--(t33);
    
    \foreach \i in {1,...,4} {
      \draw[enull] (t\i1)--(t\i2)--(t\i3)--(t\i4)--(t\i5)--(t\i6);
      \draw[enull] (f\i1)--(f\i2)--(f\i3)--(f\i4)--(f\i5)--(f\i6);
    }
    \draw[enull] (t36) to[out=-40,in=40] (f36);
    \draw[enull] (t46) to[out=-40,in=40] (f46);
    \draw[enull] (t31) to[out=-130,in=130] (f31);
    \draw[enull] (t41) to[out=-130,in=130] (f41);
    \draw[enull] (t11) to[out=-40,in=40] (f11);
    \draw[enull] (t21) to[out=-40,in=40] (f21);
    \draw[enull] (t16) to[out=-130,in=130] (f16);
    \draw[enull] (t26) to[out=-130,in=130] (f26);
    \draw[enull] (t16)--(f21)--(f16)--(t21)--(t16);
    \draw[enull] (t36)--(f41)--(f36)--(t41)--(t36);
    
    \draw[enull] (t26) to[out=-60,in=60] (f31);
    \draw[tour] (f26) to[out=-70,in=70] (f31);
    \draw[enull] (f26)--(t31);
    \draw[enull] (t31) to[out=70,in=-70] (t26);
    
    \draw[tour] (f46) to[out=180-60,in=180+60] (t11);
    \draw[enull] (f46) to[out=180-70,in=180+70] (f11);
    \draw[enull] (f11)--(t46);
    \draw[enull] (t46) to[out=180-70,in=180+70] (t11);

    \foreach \j in {1} {
      \foreach \i in {1,3,5} {
        \draw[tour] (t\j\i) -- (z\j\i);
        \draw[tour] (z\j\i) -- (f\j\i);
      }
      \foreach \i in {2,4,6} {
        \draw[tour] (f\j\i) -- (z\j\i);
        \draw[tour] (z\j\i) -- (t\j\i);
      }
    }
    \foreach \j in {2,3,4} {
      \foreach \i in {2,4,6} {
        \draw[tour] (t\j\i) -- (z\j\i);
        \draw[tour] (z\j\i) -- (f\j\i);
      }
      \foreach \i in {1,3,5} {
        \draw[tour] (f\j\i) -- (z\j\i);
        \draw[tour] (z\j\i) -- (t\j\i);
      }
    }

    \draw[tour] (f11)--(f12);
    \draw[tour] (f15)--(f16);
    \draw[tour] (t12)--(t13);
    \draw[tour] (t14)--(t15);

    \draw[tour] (t21)--(t22);
    \draw[tour] (t23)--(t24);
    \draw[tour] (t25)--(t26);
    \draw[tour] (f22)--(f23);
    \draw[tour] (f24)--(f25);

    \draw[tour] (t31)--(t32);
    \draw[tour] (t35)--(t36);
    \draw[tour] (f32)--(f33);
    \draw[tour] (f34)--(f35);

    \draw[tour] (t41)--(t42);
    \draw[tour] (t43)--(t44);
    \draw[tour] (t45)--(t46);
    \draw[tour] (f42)--(f43);
    \draw[tour] (f44)--(f45);

    \draw[tour] (t16) -- (f21);
    \draw[tour] (f36)-- (f41);
    \end{tikzpicture}
    \caption{A minimum\hyp weight
      closed walk for the graph in \cref{fig:constr}
      first visits each edge in~$E_0$ exactly once
      (the dotted edges in \cref{fig:constr}),
      and then follows the arrows as shown in this figure.
      The walk corresponds to~$x_1=0$ and~$x_2=x_3=x_4=1$.}
\label{fig:tour}
\end{figure*}

\begin{construction}\label[construction]{constr:3satred}
  Let $\varphi$~be an instance of 3-SAT.
  First,
  delete each clause containing both~$x$ and~$\bar x$:
  they are always satisfied.
  Consider now the variables~$x_1,\dots,x_n$
  and clauses~$C_1,\dots,C_m$.
  For each~$i\in\{1,\dots,n\}$,
  let $\mul i$ denote the number of
  clauses containing either~$x_i$ or~$\bar x_i$.
  We describe an instance~$I_{\varphi}=(G,\w,\prt,\ord)$ of \hcpc{}.
  Each edge will have weight one.

  Graph~$G$ contains
  a path~$(c_j^1,c^*,c_j^2)$
  for each~$j\in\{1,\dots,m\}$
  and,
  for each $i\in\{1,\dots,n\}$ and $\ell\in\{1,\dots,6\mul i\}$,
  a path \(P_{i}^\ell:=(t_i^\ell,z_i^\ell,f_i^\ell).\)
  For each~$i\in\{1,\dots,n\}$,
  it contains a cycle
  \[X_i:=(t_i^1,t_i^2,\dots,t_i^{6\mul i-1},t_i^{6\mul i},f_i^{6\mul i},f_i^{6\mul i-1},\dots,f_i^2,f_i^1,t_i^1).\]
  For each $i\in\{1,\dots,n\}$ and $i'=i\bmod n + 1$,
  cycles~$X_i$ and $X_{i'}$
  are connected to each other via a cycle
  \[
    Y_{ii'}=(t_i^{6\mul i},f_{i'}^{1},f_{i}^{6\mul i},t_{i'}^1,t_i^{6\mul i}).
  \]
  For each literal~$x_i$ in a clause~$C_j$,
  graph~$G$ contains a cycle
  \[
    Z_{ij}:=(t_i^{6\ell-3},c_j^1,a_{ij},c_j^2,t_i^{6\ell-2},b_{ij},t_i^{6\ell-3}),
  \]
  where~$\ell\leq\mul i$ is such that~$C_j$
  is the $\ell$-th clause containing~$x_i$ or~$\bar x_i$.
  For each literal~$\bar x_i$
  in a clause~$C_j$,
  graph~$G$ contains a cycle
  \[
    \bar Z_{ij}:=(f_i^{6\ell-3},c_j^1,a_{ij},c_j^2,f_i^{6\ell-2},b_{ij},f_i^{6\ell-3}),
  \]
  where~$\ell\leq\mul i$ is such that~$C_j$
  is the $\ell$-th clause containing~$x_i$ or~$\bar x_i$.
  
  The edges are ordered as follows.
  For each $i\in\{1,\dots,n\}$ and ${\ell}\in\{1,\dots,6\mul i\}$,
  $E_i^{\ell}=E(P_i^{\ell})$ is a connected class.
  They are lexicographically ordered, that is,
  \[
  E_i^{\ell}\ord E_{i'}^{{\ell}'}\iff (i< i')\vee (i=i'\wedge {\ell}\leq {\ell}').
  \]
  They are preceded by the connected class
  \[
    E_0=\bigcup_{i=1}^n E(X_i)\cup \bigcup_{i=1}^nE(Y_{i,i\bmod n + 1})\cup
    \bigcup_{x_i\in C_j}E(Z_{ij})\cup\bigcup_{\bar x_i\in C_j}E(\bar Z_{ij}),
  \]
  Finally,
  the edge set~$E^*$
  consisting of all edges incident to~$c^*$
  forms a connected
  class incomparable to all other classes.\qed
\end{construction}

\noindent
For convenience,
we collect the vertices of~$I_\varphi$
of the form~$t_i^{\ell}$ and~$f_i^{\ell}$
and of the form~$c_j^1$ and~$c_j^2$ into sets
\begin{align*}
  \vft&:=\bigcup_{i=1}^n\bigcup_{{\ell}=1}^{6\mul i}\,\{t_i^{\ell},f_i^{\ell}\}\quad\text{and}
  &
     \vc&:=\bigcup_{j=1}^m\,\{c_j^1,c_j^2\}.
\end{align*}

\begin{observation}
  \label[observation]{obs:struct}
  Let $I_\varphi=(G,\w,\prt,\ord)$~be
  the \hcpc{} instance
  constructed by \cref{constr:3satred}
  from a 3-SAT instance~$\varphi$ with $n$~variables and $m$~clauses.
  Then,
  \begin{compactenum}[(i)]
  \item\label{E0euler} the subgraph~$\sg{G}{E_0}$ is Eulerian:
    it is connected and the union of pairwise
    edge-disjoint cycles,
    that is, balanced,
    
  \item\label{imbavfc} the imbalanced vertices of~$G$
    are therefore $\vft\cup\vc$, and
  \item\label{numverts} the number of vertices, edges, and classes is $O(n+m)$,
    since $\sum_{i=1}^n\mul i\leq 3m$ in any 3-SAT formula.
  \end{compactenum}
\end{observation}

\noindent
In the following,
we will show that the \hcpc{} instance~$I_\varphi$
allows for a \emph{\chp{}}
tour if and only if $\varphi$ is satisfiable:

\begin{definition}
  A feasible solution
  for an \hcpc{} instance~$(G,\w,\allowbreak \prt,{\ord})$
  with $G=(V,E)$
  is a \emph{\chp{} tour}
  if it has weight at most $|E|+b/2$,
  where $b$~is the number of imbalanced vertices in~$G$.
\end{definition}

\begin{proposition}
  \label[proposition]{prop:equivalence}
  The \hcpc{} instance~$I_\varphi=(G,\w,\prt,\ord)$
  created from a 3-SAT instance~$\varphi$
  by \cref{constr:3satred}
  allows for a \chp{} tour if and only if~$\varphi$ is
  satisfiable.
\end{proposition}
\noindent
In the rest of this section,
it remains to prove \cref{prop:equivalence},
which, together with \cref{obs:struct}\eqref{numverts},
yields \cref{thm:hard}.

\paragraph{Satisfiability of~$\varphi$ implies a tight tour in~$I_\varphi$}
Assume that~$\varphi$ is satisfiable.
We show a \chp{} tour~$T$ for~$I_\varphi$.
Without loss of generality,
assume that~$x_1$ is ``true'':
otherwise, we can replace~$x_1$ by~$\bar x_1$
throughout the formula~$\varphi$.

The \chp{} tour~$T$ for~$I_\varphi$ then looks as follows
(an example is shown in \cref{fig:tour}).
It starts in~$f_1^1$,
first visits each edge of~$E_0$ exactly once
and returns to~$f_1^1$.
This is possible by \cref{obs:struct}\eqref{E0euler}.
Then,
it remains to traverse the paths~$(t_i^\ell,z_i^\ell,f_i^\ell)$
for each~$i\in\{1,\dots,n\}$ and~$\ell\in\{1,\dots,6\mul i\}$
and the paths~$(c_j^1,c^*,c_j^2)$ for each~$j\in\{1,\dots,m\}$.
This is done as follows.
For~$i$ from~$1$ to~$n$,
if $x_i$~is ``true'',
$T$~visits the vertices \[f_i^1,z_i^1,t_i^1,~~t_i^2,z_i^2,f_i^2,~~f_i^3,z_i^3,t_i^3,~~t_i^4,z_i^4,f_i^4~~\dots~~,t_i^{6\mul i},z_i^{6\mul i},f_i^{6\mul i},\]
for some $\ell\in\{1,\dots,\mul i\}$ taking
a detour through the vertices $t_i^{6\ell-3},c_j^1,c^*,c_j^2,t_i^{t\ell-2}$
if clause~$C_j$ contains~$x_i$
and
$(c_j^1,c^*,c_j^2)$~has not been traversed before.
If $x_i$~is ``false'',
then $T$~visits
\[t_i^1,z_i^1,f_i^1,~~f_i^2,z_i^2,t_i^2,~~t_i^3,z_i^3,f_i^3,~~f_i^4,z_i^4,t_i^4,~~\dots~~,f_i^{6\mul i},z_i^{6\mul i},t_i^{6\mul i},\]
for some $\ell\in\{1,\dots,\mul i\}$ taking
a detour through the vertices $f_i^{6\ell-3},c_j^1,c^*,c_j^2,f_i^{6\ell-2}$
if clause~$C_j$ contains~$\bar x_i$
and
$(c_j^1,c^*,c_j^2)$~has not been traversed before.
Finally,
after~$f_n^{6\mul n}$ or~$t_n^{6\mul n}$,
the walk~$T$ returns to~$f_1^1$.
Note that this traversal
is possible
due to the cycle~$Y_{i,i\bmod n+1}$
for each~$i\in\{1,\dots,n\}$.%

Observe that the closed walk~$T$ contains
all edges and respects precedence constraints:
For the edges in~$E_0$ and
all paths~$(t_i^\ell,z_i^\ell,f_i^\ell)$
for $i\in\{1,\dots,n\}$ and~$\ell\in\{1,\dots,6\mul i\}$,
this is obvious.
To see that the path~$(c_j^1,c^*,c_j^2)$
has been traversed for each~$j\in\{1,\dots,m\}$,
observe that each clause~$C_j$ contains a true literal,
so that a detour via~$c_j^1,c^*,c_j^2$ is taken.

To see that $T$~is \chp{},
we check which edges
are traversed a second time.
When $x_i$ is ``true'',
the edges~$\{f_i^{2\ell},f_i^{2\ell+1}\}\in E_0$ for~$\ell\in\{1,\dots,\mul i-1\}$
are visited a second time,
whereas
each edge~$\{t_i^{2\ell-1},t_i^{2\ell}\}\in E_0$ for~$\ell\in\{1,\dots,\mul i\}$
is traversed a second time
or skipped by a detour that traverses the edges~$\{t_i^{2\ell-1},c_j^1\}$ and~$\{t_i^{2\ell},c_j^2\}$
a second time.
Analogously
when $x_i$ is ``false''.
Moreover,
for each $i\in\{1,\dots,n\}$,
one edge of the cycle~$Y_{i,i\bmod n+1}$ is visited a second time:
it joins the last vertex visited by~$T$ in~$X_i$
to the first vertex visited by~$T$ in~$X_{i\bmod n+1}$.
We thus see that
the edges visited a second time form a matching.
Their endpoints are the $b$~imbalanced vertices~$\vft\cup\vc$.
Thus,
$T$~traverses not more than $|E|+b/2$ edges.

\paragraph{\Chp{} tour for $I_\varphi$ implies satisfiability of~$\varphi$}
Assume that~$I_\varphi$
allows for a \chp{} tour~$T$.
We show that
$\varphi$ is satisfiable.

\begin{lemma}
  \label[lemma]{lem:decomp}
  Let $M=E(T)\setminus E$,
  that is,
  $M$~is the multiset of
  the edges that the tight tour $T$~traverses additionally to~$E$
  (taking into account the multiplicity of additional visits).
  Then,
  \begin{compactenum}[(i)]
    
  \item\label{matching} $M\subseteq E_0$ is a perfect matching
    on the vertices~$\vft\cup \vc{}$,
    in particular,
    $M$~contains each edge at most once,
    
  \item\label{matchin2} each edge in~$M$
    has an endpoint in~$\vft{}$,
  \item\label{t2match} $\sg{G}{(E\setminus E_0)\cup M}$ is connected.
  \end{compactenum}
\end{lemma}

\begin{proof}
  \eqref{matching}
  Since $T$~is a closed walk,
  all vertices in~$\sg{G}{T}$ are balanced,
  whereas its subgraph~$\sg{G}{E}=G$
  has $b$~imbalanced vertices. 
  Since $T$~contains at most $|E|+b/2$~edges,
  the graph~$\sg{G}{T}$ contains at most $b/2$~edges
  additionally to those in~$\sg{G}{E}$.
  Thus,
  $\sg{G}{T}$~contains a set~$M$ of at most~$b/2$ edges
  whose endpoints are the $b$~imbalanced vertices of~$\sg{G}{E}$.
  By \cref{obs:struct}\eqref{imbavfc},
  these are exactly the vertices~$\vft\cup\vc$.
  This is only possible if $M$~is a perfect matching on~$\vft\cup\vc$.
  Since each edge of~$G$ that is not in~$E_0$
  has at least one balanced endpoint
  (namely,  $c^*$~or one of the $z_i^\ell$),
  we easily get~$M\subseteq E_0$.
  
  \eqref{matchin2}
  The only edges in~$E_0$
  that have no endpoints in~$\vft$
  have one of the vertices
  of the form~$a_{ij}$ or~$b_{ij}$ as endpoints.
  Since these
  are only on the cycle~$Z_{ij}$ or~$\bar Z_{ij}$,
  they are balanced.
  Thus,
  $M$~cannot contain such edges.

  \eqref{t2match}
  Let $T^*$~be a tight tour for~$(G,\w,\prt,\ord)$
  such that
  $T^*$ visits exactly the edges in~$M$ twice and
  the minimal prefix~$T_1$ of~$T^*$
  traversing all edges in~$E_0$
  has minimum length
  (the tight tour~$T$ and \eqref{matching} witness
  the existence of~$T^*$).
  Let $T_2$~be the rest of~$T^*$.
  We will prove that
  $T_1$~is an Euler tour for~$\sg{G}{E_0}$.
  Then \eqref{t2match} follows
  since $\sg{G}{(E\setminus E_0)\cup M}$ is even Eulerian:
  $T_2$~visits all edges in~$(E\setminus E_0)\cup M$,
  since they are not visited by~$T_1$;
  $T_2$ is closed since~$T^*$ and its prefix~$T_1$ are;
  and,
  by choice of~$T^*$,
  $T_2$~does not visit any edge
  in~$(E\setminus E_0)\cup M$ more than once.
  It remains to prove that $T_1$~is indeed
  an Euler tour for $\sg{G}{E_0}$.

  We first prove that $T_1$~is a closed walk.
  By the minimality of~$T_1$
  and choice of~$\ord$,
  $T_1$~does not end in~$c^*$
  and
  does not contain edges from any class~$E_i^\ell$.
  It might contain edges from~$E^*$.
  Assume,
  for the sake of a contradiction,
  that
  $T_1$~starts at some vertex~$s$ and ends at some vertex~$t\ne s$.
  Then, $t$~is not balanced in $\sg{G}{T_1}$ but
  balanced in $\sg{G}{E_0}$ by \cref{obs:struct}\eqref{E0euler}.
  Thus, there is an edge~$e=\{t',t\}\in E^* \cup M$ on~$T_1$.
  The graph $\sg{G}{E(T_1)\setminus \{e\}}$
  contains~$E_0$,
  is connected,
  all its vertices except for~$s$ and~$t'$ are balanced,
  and it therefore has an Euler walk~$T_1'$.
  It follows that
  $(T_1', e, T_2)$~is another tight tour for~$G$
  visiting all edges with the same multiplicity as~$T^*$,
  yet its prefix~$T_1'$ containing~$E_0$ satisfies $|T_1'|<|T_1|$.
  This contradicts the choice of~$T^*$.
  \looseness=-1
  We now show that~$T_1$
  traverses each edge~$e\in E_0\cup E^*$
  at most once.
  Towards a contradiction,
  assume that it traverses $e=\{u,v\}$ twice.
  Then, $v\in\vft$ by \eqref{matchin2}.
  Thus,
  $v$~is not incident to any edges in~$E^*$
  and,
  because
  $v$~is balanced in $\sg{G}{E_0}$
  by \cref{obs:struct}\eqref{E0euler},
  it is also balanced in~$\sg{G}{E_0\cup E^*}$.
  Since $v$~is balanced in~$\sg{G}{E_0\cup E^*}$,
  balanced in~$\sg{G}{T_1}$,
  and $T_1$~traverses~$e$ twice,
  $T_1$~also traverses another edge incident to~$v$ twice,
  contradicting~\eqref{matching}.
  
  Finally,
  we prove that $T_1$~contains
  \emph{only} edges of~$E_0$.
  Towards a contradiction,
  assume that
  $T_1$~contains an edge~$\{c^*,c\}\in E^*$.
  Vertex~$c\in\vc{}$ is balanced in~$\sg{G}{T_1}$,
  yet not balanced in its subgraph~$\sg{G}{E_0\cup \{c^*,c\}}$
  by \cref{obs:struct}\eqref{E0euler}.
  Thus,
  $\sg{G}{T_1}$~contains
  some edge~$e\in E_0\cup E^*$ twice,
  which is impossible.
\end{proof}

\begin{figure*}
  \centering
  \begin{subfigure}{0.45\textwidth}
    \begin{tikzpicture}[x=1.5cm]
      \foreach \i in {1,...,5}
      {
        \node (t\i) at (6-\i,0) {$t_i^{6\ell-\i}$};
        \node (f\i) at (6-\i,-1) {$f_i^{6\ell-\i}$};
      }
      \node (t0) at (6,0) {$t_i^{6\ell}$};
      \node (f0) at (6,-1) {$f_i^{6\ell}$};
      \node (c1) at (4,1) {$c_j^2$};
      \node (c2) at (3,1) {$c_j^1$};
      
      \draw[enull] (c1)--(t2);
      \draw[enull] (c2)--(t3);

      \draw[enull] (t0)--(f0);
      \draw[enull] (t5)--(f5);

      \draw[enull] (t1)--(t0);
      \draw[matching] (t1)--(t2);
      \draw[enull] (t3)--(t2);
      \draw[matching] (t3)--(t4);
      \draw[enull] (t5)--(t4);
      
      \draw[matching] (f1)--(f0);
      \draw[enull] (f1)--(f2);
      \draw[matching] (f3)--(f2);
      \draw[enull] (f3)--(f4);
      \draw[matching] (f5)--(f4);
    \end{tikzpicture}
    \caption{Case 1: $\{f_i^{6\ell-3},f_i^{6\ell-2}\}$ is covered.}
    \label{matcha}
  \end{subfigure}
  \hfill
    \begin{subfigure}{0.45\textwidth}
    \begin{tikzpicture}[x=1.5cm]
      \foreach \i in {1,...,5}
      {
        \node (t\i) at (6-\i,0) {$t_i^{6\ell-\i}$};
        \node (f\i) at (6-\i,-1) {$f_i^{6\ell-\i}$};
      }
      \node (t0) at (6,0) {$t_i^{6\ell}$};
      \node (f0) at (6,-1) {$f_i^{6\ell}$};
      \node (c1) at (4,1) {$c_j^2$};
      \node (c2) at (3,1) {$c_j^1$};
      
      \draw[enull] (c1)--(t2);
      \draw[enull] (c2)--(t3);

      \draw[enull] (t0)--(f0);
      \draw[enull] (t5)--(f5);

      \draw[matching] (t1)--(t0);
      \draw[enull] (t1)--(t2);
      \draw[enull] (t3)--(t2);
      \draw[enull] (t3)--(t4);
      \draw[matching] (t5)--(t4);
      
      \draw[enull] (f1)--(f0);
      \draw[matching] (f1)--(f2);
      \draw[enull] (f3)--(f2);
      \draw[matching] (f3)--(f4);
      \draw[enull] (f5)--(f4);
    \end{tikzpicture}
    \caption{Case 2: $\{f_i^{6\ell-3},f_i^{6\ell-2}\}$ is not covered.}
    \label{matchb}
  \end{subfigure}
  \caption{The two cases in the proof of \cref{lem:covering}.
    Dotted edges are all the possibly present edges in~$E_0$
    available for inclusion in~$M$
    by \cref{lem:decomp}\eqref{matching}
    (the edge $\{t_i^{6\ell-5},f_i^{6\ell-5}\}$
    is present if~$\ell=1$,
    the edge $\{t_i^{6\ell},f_i^{6\ell}\}$
    is present if~$\ell=\mul i$ ).
    Including $\{f_i^{6\ell-3},f_i^{6\ell-2}\}$ in~$M$
    or excluding it from~$M$
    force all the bold edges into~$M$
    due to the fact that all vertices
    must be contained in some edge of~$M$
    and that edges
    drawn above each other cannot both be part of~$M$.
}
\label{fig:covering}
\end{figure*}
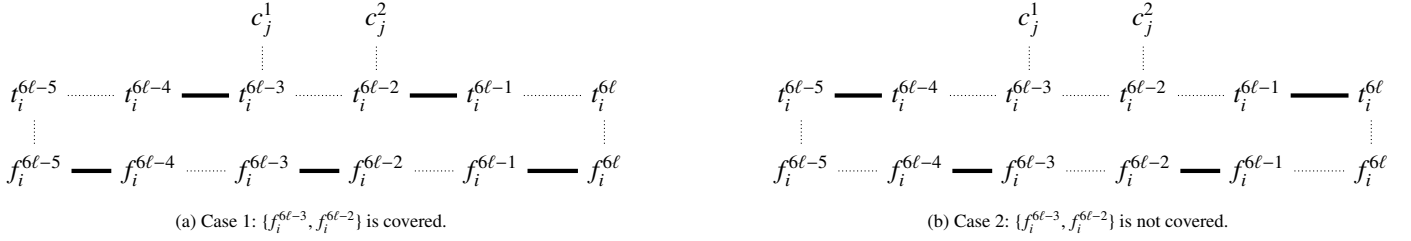

\pagebreak[3]
\noindent
We now show that the matching~$M$
from \cref{lem:decomp}
takes one of two possible forms in each
variable cycle~$X_{i}$.
This will correspond to setting a variable to ``true'' or ``false''.

\begin{definition}
  Let $i\in\{1,\dots,n\}$ and~$\ell\in \{1,\dots,\mul i\}$.
  We call an edge~$\{t_i^{6\ell-3},t_i^{6\ell-2}\}$ \emph{covered}
  if $\{t_i^{6\ell-3},t_i^{6\ell-2}\}\in M$
  or if there is a~$j\in\{1,\dots,m\}$
  such that both~$\{t_1^{6\ell-3},c_j^1\}$ and $\{t_1^{6\ell-2},c_j^2\}$
  are in~$M$.
  
  We call an edge~$\{f_i^{6\ell-3},f_i^{6\ell-2}\}$ \emph{covered}
  if $\{f_i^{6\ell-3},f_i^{6\ell-2}\}\in M$
  or if there is a~$j\in\{1,\dots,m\}$
  such that both~$\{f_1^{6\ell-3},c_j^1\}$ and $\{f_1^{6\ell-2},c_j^2\}$
  are in~$M$.
\end{definition}

\begin{lemma}
  \label[lemma]{lem:covering}
  For each $i\in\{1,\dots,n\}$,
  either all $\{t_i^{6\ell-3},t_i^{6\ell-2}\}$ are covered
  or all $\{f_i^{6\ell-3},f_i^{6\ell-2}\}$ are covered for
  $\ell\in \{1,\dots,\mul i\}$.
\end{lemma}

\begin{proof}
\looseness=-1
  For any~$i\in\{1,\dots,n\}$
  and~$\ell\in\{1,\dots,\mul i\}$,
  we first show that exactly one of
  $\{t_i^{6\ell-2},t_i^{6\ell-3}\}$ and
  $\{f_i^{6\ell-2},f_i^{6\ell-3}\}$
  is covered.
  Note that, by \cref{constr:3satred},
  at most one of these pairs of vertices
  is attached to~$\{c_j^1,c_j^2\}$ for any~$j\in\{1,\dots,m\}$.
  Without loss of generality,
  let this be $\{t_i^{6\ell-2},t_i^{6\ell-3}\}$.
  The other case is symmetric.

  Denote~$R:=E\setminus E_0$.
  By \cref{lem:decomp}\eqref{t2match},
  $\sg{G}{R\cup M}$ is connected.
  Thus,
  there is at least one edge of~$M$
  leaving any subset of connected components of~$\sg{G}{R}$
  and,
  for each~$h\in\{6\ell-5,\dots,6\ell-1\}$,
  only one of $\{t_i^h,t_i^{h+1}\}$ and
  $\{f_i^h,f_i^{h+1}\}$ is in~$M$:
  otherwise,
  the matching~$M$ could not contain any edge
  leaving the set of connected
  components $\{\{t_i^h,z_i^h,f_i^h\},\{t_i^{h+1},z_i^{h+1},f_i^{h+1}\}\}$.
  We also exploit
  that,
  by \cref{lem:decomp}\eqref{matching},
  all vertices in~$\vft$ must be incident to an edge of~$M$.

  \looseness=-1
  We now distinguish two cases,
  illustrated in \cref{fig:covering}.
  First,
  assume that $\{f_i^{6\ell-3},f_i^{6\ell-2}\}$ is covered,
  that is, in~$M$.
  Then all bold edges shown in \cref{matcha}
  are in~$M$.
  Thus,
  the edge  $\{t_i^{6\ell-2},t_i^{6\ell-3}\}$ is not covered.
  If
  the edge
  $\{f_i^{6\ell-3},f_i^{6\ell-2}\}$ is \emph{not} covered,
  that is, not in~$M$,
  then all bold edges shown in \cref{matchb} are in~$M$.
  To match the vertices~$t_i^{6\ell-3}$ and~$t_i^{6\ell-2}$,
  one either has $\{t_i^{6\ell-3},t_i^{6\ell-2}\}\in M$
  or $\{\{t_i^{6\ell-3},c_j^1\},\{t_i^{6\ell-2},c_j^2\}\}\subseteq M$.
  That is,
  $\{t_i^{6\ell-3},t_i^{6\ell-2}\}$ is covered.

  Finally,
  towards a contradiction,
  assume that there are~$\ell,\ell'$
  such that
  $\{f_i^{6\ell-2},f_i^{6\ell-3}\}$
  and
  $\{t_i^{6\ell'-2},t_i^{6\ell'-3}\}$
  are covered.
  Then
  we can choose~$\ell,\ell'$ so that~$|\ell-\ell'|=1$.
  Assume $\ell'=\ell+1$,
  the other case is symmetric.
  Then,
  as illustrated in \cref{matcha},
  vertex~$t_i^{6\ell}$ has to be matched to~$t_i^{6\ell'-5}$
  (there is no edge~$\{t_i^{6\ell},f_i^{6\ell}\}$ in this case
  by \cref{constr:3satred},
  since $\ell<\mul i$).
  However,
  vertex~$t_i^{6\ell'-5}$ is already matched to~$t_i^{6\ell'-4}$,
  so that this is impossible.
\end{proof}

\noindent
We can now easily prove
that, since $I_\varphi$ has a \chp{} tour~$T$,
the formula~$\varphi$ is satisfiable,
thus concluding the proof of \cref{prop:equivalence}.
By \cref{lem:decomp}\eqref{matching} and \eqref{matchin2},
for each clause~$C_j$ of~$\varphi$,
the vertices~$c_j^1$ and~$c_j^2$
are matched to vertices in~$\vft$ by~$M$.
By \cref{constr:3satred},
$c_j^1$ can only be matched to~$t_i^{6\ell-3}$
or~$f_i^{6\ell-3}$ for some~$i\in\{1,\dots,n\}$ and~$\ell\in\{1,\dots,\mul i\}$.
By \cref{lem:covering},
if $c_j^1$ is matched to~$t_i^{6\ell-3}$,
then $t_i^{6\ell-2}$ is matched to~$c_j^2$
and the edges $\{t_i^{6\ell-3},t_i^{6\ell-2}\}$ are covered
for \emph{all}~$\ell\in\{1,\dots,\mul i\}$,
whereas $\{f_i^{6\ell-3},f_i^{6\ell-2}\}$
is \emph{not} covered for \emph{any}~$\ell\in\{1,\dots,\mul i\}$.
Thus,
clause~$C_j$
(and all other clauses containing~$x_i$)
can be satisfied by
setting variable~$x_i$ to~``true''.
If, on the other hand,
$c_j^1$ is matched to~$f_i^{6\ell-3}$,
then $f_i^{6\ell-2}$ is matched to~$c_j^2$
and the edges $\{f_i^{6\ell-3},f_i^{6\ell-2}\}$ are covered
for \emph{all}~$\ell\in\{1,\dots,\mul i\}$,
whereas $\{t_i^{6\ell-3},t_i^{6\ell-2}\}$
is \emph{not} covered for \emph{any}~$\ell\in\{1,\dots,\mul i\}$.
Thus,
clause~$C_j$
(and all other clauses containing~$\bar x_i$)
can be satisfied by setting~$x_i$ to~``false''.

\section{Relation between \hcpl{} and
  the Rural Postman}
\noindent
\citet{DST87} showed
how to reduce \hcpcl{}
to polynomial\hyp time solvable
special cases of the following problem.

\label{sec:RPP}
\begin{problem}[\stRPPlong{}, \stRPP]\label[problem]{prob:stRPP}
  \label{prob:strpp}
  \leavevmode
  \begin{compactdesc}  
  \item[Input:] An undirected graph~$G=(V,E)$, 
    edge weights~$\w\colon E\to \N$, vertices 
    $s,t\in V$, and a subset~$R \subseteq E$ of \emph{required edges}.
    
  \item[Find:] A walk~$W^*$ of minimum total weight~$\w(W^*)$,
    starting in~$s$, ending in~$t$, and traversing all edges of~$R$.
  \end{compactdesc}
\end{problem}
\noindent
In general,
\stRPP{} is strongly NP\hyp hard,
as well as the better known
Rural Postman Problem (\RPP), where the goal 
is to find a \emph{closed} walk~\cite{LR76}.
\citet{DST87} reduce \hcpcl{}
to multiple \stRPP{} instances
in which the subgraph~$\sg{G}{R}$ is connected.
Since this case of \stRPP{}
is polynomially\hyp time solvable,
this yields a polynomial\hyp time algorithm for \hcpcl{}  \cite{DST87}.

We now show that,
while applying the same construction to \hcpl{} does not
yield polynomial\hyp time solvable instances of \stRPP{},
it allows to transfer
running times,
approximation factors,
and error probabilities
of \stRPP{} algorithms to \hcpl{}.
This is in contrast
to a construction due to \citet{CGGL04},
who showed a polynomial\hyp time reduction
of \hcpl{} to \RPP{}
that does \emph{not} allow to transfer
  approximation factors:
  it introduces very heavy required edges,
  which always contribute to the goal function
  and thus make bad approximate solutions ``look'' good.
We now describe the construction of \citet{DST87}. %

	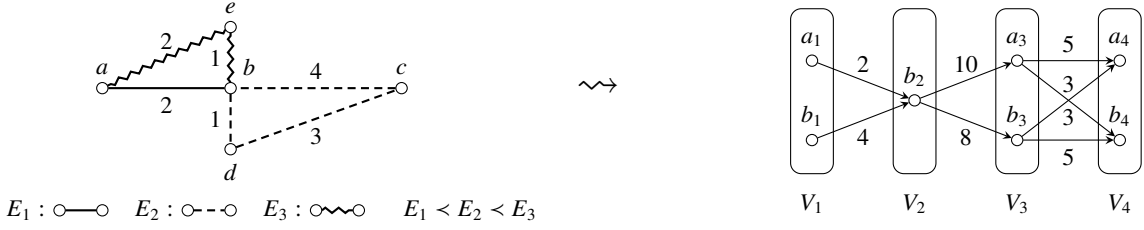
\begin{figure*}
	\centering
	
	\begin{tikzpicture}[scale=0.9, every node/.style={scale=0.9}]
	
	\def\xr{1.25}
	\def\yr{0.9}
	\tikzstyle{xnode}+=[circle, scale=1/2, draw]
	
	\tikzstyle{edge1}+=[-,thick]
	\tikzstyle{edge2}+=[-,thick, densely dashed]
	\tikzstyle{edge3}+=[snake=zigzag, segment amplitude=1pt, segment length=5pt, thick]
	\tikzstyle{arrow} = [->,>=stealth]

	\begin{scope}[xshift=-2*\xr cm]
	
	\node (e) at (0*\xr,3*\yr)[xnode,label=90:{$e$}]{};
	\node (a) at (-1.5*\xr,2*\yr)[xnode,label=90:{$a$}]{};
	\node (b) at (0*\xr,2*\yr)[xnode,label=45:{$b$}]{};
	\node (c) at (2*\xr,2*\yr)[xnode,label=90:{$c$}]{};
	\node (d) at (0,\yr)[xnode,label=-90:{$d$}]{};
	
	\draw[edge1] (a) -- node[midway,below]{2} (b);
	\draw[edge2] (b) -- node[left]{1} (d) --node[midway,below]{3} (c) --node[midway,above]{4} (b);
	\draw[edge3] (b) -- node[left]{1} (e) -- node[midway,above]{2} (a);
	
	\end{scope}	

	\begin{scope}[xshift=-4*\xr cm, yshift=0*\yr cm]
	
	\def\edlen{0.5*\xr}
	
	\node(e1) at (0.5*\xr - \edlen,0)[xnode,label=180:{$E_1:$}]{};
	\node(e11) at (0.5*\xr,0)[xnode]{};
	\draw[edge1] (e1) -- (e11);
	\node(e2) at (2*\xr - \edlen,0)[xnode,label=180:{$E_2:$}]{};
	\node(e22) at (2*\xr,0)[xnode]{};
	\draw[edge2] (e2) -- (e22);
	\node(e3) at (3.5*\xr-\edlen,0)[xnode,label=180:{$E_3:$}]{};
	\node(e33) at (3.5*\xr,0)[xnode]{};
	\draw[edge3] (e3) -- (e33);
	\node at (4.8*\xr,0)[]{$E_1 \ord E_2 \ord E_3$};
	
	\end{scope}
	
	\node at (2.3*\xr,2*\yr)[scale=1.8]{$\leadsto$};
	
	\begin{scope}[xshift=4.8*\xr cm]
	
	\def\xset{0.25*\xr} %
	\def\ldist{0.6*\xr} %
	\def\hdist{0.65*\yr} %
	
	\draw[rounded corners] (-\xset, 3.3*\yr) rectangle (\xset,0.6*\yr) node[below,xshift=-9pt, yshift=-5pt]{$V_1$};

	\node (v11) at (0, 1.8*\yr + \hdist)[xnode,label=90:{$a_1$}]{};
	\node (v12) at (0, 1.8*\yr - \hdist)[xnode,label=90:{$b_1$}]{};

	\draw[rounded corners] (-\xset+2*\ldist,3.3*\yr) rectangle (\xset+2*\ldist, 0.6*\yr) node[below,xshift=-9pt, yshift=-5pt]{$V_2$};

	\node (v23) at (2*\ldist,1.8*\yr)[xnode,label=90:{$b_2$}]{};

	\draw[rounded corners] (-\xset+4*\ldist,3.3*\yr) rectangle (\xset+4*\ldist,0.6*\yr) node[below,xshift=-9pt, yshift=-5pt]{$V_3$};

	\node (v32) at (4*\ldist,1.8*\yr + \hdist)[xnode,label=90:{$a_3$}]{};
	\node (v33) at (4*\ldist,1.8*\yr - \hdist)[xnode,label=90:{$b_3$}]{};

	\draw[rounded corners] (-\xset+6*\ldist,3.3*\yr) rectangle (\xset+6*\ldist,0.6*\yr) node[below,xshift=-9pt, yshift=-5pt]{$V_4$};
	
	\node (v41) at (6*\ldist,1.8*\yr + \hdist)[xnode,label=90:{$a_4$}]{};
	\node (v42) at (6*\ldist,1.8*\yr - \hdist)[xnode,label=90:{$b_4$}]{};

\draw[arrow] (v11) -- node[midway,above]{2} (v23);
\draw[arrow] (v12) -- node[midway,below]{4} (v23);
\draw[arrow] (v23) -- node[midway,above]{10}(v32);
\draw[arrow] (v23) -- node[midway,below]{8}(v33);
\draw[arrow] (v33) -- node[midway,below]{3}(v41);
\draw[arrow] (v33) -- node[midway,below]{5}(v42);
\draw[arrow] (v32) -- node[midway,above]{5}(v41);
\draw[arrow] (v32) -- node[midway,above]{3}(v42);

\end{scope}

\end{tikzpicture}
\caption{Illustration for \cref{constr:Gstar}: from a graph~$G$ 
  with $k=3$~edge classes (on the left),
  \cref{constr:Gstar} constructs
  a graph~$\Gamma$ with $k+1$~layers (on the right).
  Note that, for example, 
  vertex $b \in V(G)$ is the only vertex in~$V(E_2)$ incident to edges 
  of previous classes, thus, its copy~$b_2 \in V(\Gamma)$
  is the only vertex in layer~$V_2$.}
\label{fig:Gstar}
\end{figure*}

\begin{definition}
  \label[definition]{def:Ruvi}
  In this section,
  we denote the edge classes of
  \hcpl{} instances~$(G,\w,\prt,\ord)$
  by~$E_1,\dots,\allowbreak E_k$,
  where $E_i \ord E_j$ if and only if $1 \le i<j \le k$. 

  By $R[u,v,i]$,
  we denote the \stRPP{} instance
  of finding a minimum\hyp weight
  walk between the
  vertices~$u$ and~$v$ in
  $\sg{G}{E_1 \cup \dots \cup E_i}$
  traversing all edges in~$E_i$.
  By $P[u,v,i]$,
  we denote an arbitrary optimal solution to~$R[u,v,i]$.
\end{definition}

\begin{construction}\label[construction]{constr:Gstar}
  From a \hcpl{} instance~$(G,\w,\prt,\ord)$,
  construct a directed arc\hyp weighted
  graph~$\Gamma=(V_{\Gamma},A_{\Gamma})$ as illustrated in \cref{fig:Gstar}:
  The vertex set~$V_{\Gamma}=\bigcup_{i=1}^{k+1} V_i$ 
  is a union of \emph{layers~$V_i$}.
  For each~$i\in\{2,\dots,k\}$,
  layer~$V_i$
  contains a copy of each vertex in~$G$
  that is incident to an edge of~$E_i$ and of
  any predecessor class.
Namely,
  \begin{align*}
    V_1&=\{ u_{1} \mid u \in V(E_1) \},
    &V_{k+1}&=\{ u_{k+1} \mid u \in V(E_1) \},\\
    V_i&=\biggl\{ u_i \biggm| u \in V(E_i)\cap \bigcup_{j=1}^{i-1} V(E_{j}) \biggr\}&\text{for }i&\in\{2,\dots,k\}.
   \end{align*}

  \noindent For each pair of vertices~$u_i \in V_{i}$ and $ v_{i+1} \in V_{i+1}$,
  where $i\in\{1,\dots,k\}$,
  there is an arc~$(u_i,v_{i+1}) \in A_{\Gamma}$
  of weight $\w_{\Gamma}(u_{i},v_{i+1})=\w(P[u,v,i])$.
  If $P[u,v,i]$~does not exist,
  there is no arc $(u_i,v_{i+1})$.
\end{construction}

\begin{proposition}[\citet{DST87}]
  \label[proposition]{prop:lpath}
  Let $I:=(G,\w,\prt,\ord)$~be an \hcpl{} instance
  and $\Gamma$~be constructed from~$I$ by
  \cref{constr:Gstar}.
  Then,
  the weight of an optimal solution to~$I$
  coincides with the %
  weight of a least-weight
  \emph{\Gpath{}} in~$\Gamma$,
  where a \emph{\Gpath{}} in~$\Gamma$ is
  a path from~$v_1 \in V_1$ to~$v_{k+1} \in V_{k+1}$
  such that $v_1$~and $v_{k+1}$~are copies of
  the same vertex $v \in V(E_1)$.

  In particular,
  each \Gpath{} in~$\Gamma$ %
  has the form~$J=(v_1, y_2^{2}, y_3^3,\allowbreak
  \dots,\allowbreak y_k^k, v_{k+1})$,
  where $y_i^i \in V_i$ for $i\in\{2,\dots,k\}$
  and
  concatenating the corresponding walks $P[v,y^2,1], 
  P[y^2,y^3,2], \allowbreak \dots,\allowbreak P[y^k,v,k]$ yields
  a feasible solution~$W_J$ of weight
  $\w(W_J) = \wG(J)$ for~$I$.
\end{proposition}

\noindent
\cref{constr:Gstar} can be used
to solve \hcpcl{} in $O(kn^5)$~time:
$\Gamma$~has at most $kn^2$~arcs,
the weight of each is computed by
solving an \stRPP{} instance~$R[u,v,i]$,
which works in $O(n^3)$~time since the set~$E_i$
of required edges is connected \citep{DST87}.
It remains to find a \Gpath{} in~$\Gamma$.
This can be done in $O(kn^3)$~time
by $n$~times calling a linear\hyp time
single\hyp source
shortest\hyp path algorithm for directed acyclic graphs.

However,
when applied to \hcpl{},
\cref{constr:Gstar}
gets to solve \stRPP{} instances~$R[u,v,i]$
where the set of required edges~$E_i$
might be disconnected.
Since we do not know how to solve them in polynomial time,
in \cref{sec:apx,sec:fpt},
we will solve them using approximation algorithms
and randomized fixed\hyp parameter algorithms.
Their performance guarantees 
carry over to \hcpl{} as follows.

\begin{lemma}\label[lemma]{thm:strppApprox}
  Let $I=(G,\w,\prt,\ord)$~be an \hcpl{} instance.
  Assume that there is an algorithm
  running in $\rpptime$~time that,
  given any \stRPP{} instance~$R[u,v,i]$
  (cf.\ \cref{def:Ruvi}),
  outputs an $\alpha$\hyp approximate solution for~$R[u,v,i]$
  with probability at least~$1-p$.

  Then, there is an algorithm running in
  $O(kn^2\rpptime + kn^3)$~time
  that returns an
  $\alpha$\hyp approximate solution for~$I$
  with probability at least $1-pk$.
\end{lemma}

\begin{proof}
  Let $\algA$~denote the assumed
  randomized approximation algorithm for
  solving \stRPP{} instances~$R[u,v,i]$.
  Since we can check
  the feasibility of any solution returned
  by $\algA$ in linear time,
  we can assume that $\algA$~makes
  only one\hyp sided errors:
  For an infeasiable instance~$R[u,v,i]$,
  it returns nothing.
  For a feasible instance~$R[u,v,i]$,
  with probability at most~$p$,
  it may return nothing or produce a solution
  that is more expensive than an $\alpha$\hyp approximate solution.
  Moreover,
  since feasibility of~$I$ is easy to check \citep{DST87},
  we will assume that $I$~has a feasible solution.
  Then we compute a solution to~$I$ as follows.

  \looseness=-1
  Construct an arc\hyp weighted
 directed graph 
  $\widetilde\Gamma=(\widetilde V_\Gamma, \widetilde A_\Gamma)$ 
  from~$G$ as described in \cref{constr:Gstar},
  yet %
  for each~$i\in\{1,\allowbreak\dots,\allowbreak k\}$
  and every~$u_i \in V_{i}$ and~$v_{i+1} \in V_{i+1}$,
  the weight
  $\wtG(u_{i},v_{i+1})=\allowbreak\w(\widetilde P[u,v,i])$,
  where $\widetilde P\bigl[u,v,i\bigl]$
  is computed by applying~\algA{}
  to the \stRPP{} instance~$R[u,v,i]$
  (if \algA{} fails to produce a solution,
  then let there be no arc~$(u_i,v_{i+1})$ in~$\widetilde\Gamma$).
  Finally,
  try to compute a least\hyp weight \Gpath{}~$J$
  in~$\widetilde\Gamma$.
  If it exists,
  then the corresponding closed walk~$W_J$
  is a feasible solution of weight $\w(W_J) = \wtG(J)$
  for~$I$.
  The running time of the
  whole procedure is~$O(kn^2\rpptime + kn^3)$
  since the graph~$\widetilde\Gamma$ has $O(kn^2)$~arcs,
  the weight of each can be computed in $\rpptime$~time,
  and the least\hyp weight \Gpath{} in~$\Gamma$
  can finally be found by $n$~times
  applying a single\hyp source
  shortest\hyp path algorithm for directed acyclic graphs.
  It remains to analyze the probability
  that the procedure
  returns an $\alpha$\hyp approximate
  solution for~$I$.

  To this end,
  let $W^*$~be an optimal solution to~$I$,
  $\Gamma=(V_\Gamma,A_\Gamma)$~be constructed by
  \cref{constr:Gstar} from~$I$,
  and $J^* = (x_1, y_2^{2},y_3^3,\ldots,\allowbreak
  y_k^k,\allowbreak  x_{k+1})$
  be a least\hyp weight \Gpath{} in~$\Gamma$.
  First,
  assume that $\algA{}$ indeed
  produced an $\alpha$\hyp approximate solution
  for each instance~$R[u,v,i]$
  corresponding to any arc~$(u_i,v_{i+1})$ on~$J^*$.
  Then,
  for each arc~$(u_i,v_{i+1})$ on~$J^*$,
  \[
    \wtG(u_{i},v_{i+1})=\w(\widetilde P[u,v,i])\leq \alpha\w(P[u,v,i])=\alpha\wG(u_{i},v_{i+1})
  \]
  and $J^*$ witnesses the existence
  of the computed
  least\hyp weight \Gpath{}~$J$ in $\widetilde\Gamma$.
  Thus,
  the weight~$\w(W_J)=\wtG(J)$ is at most
  \begin{align*}
            \wtG(J^*)
            &= \wtG(x_1,y_2^{2}) +\wtG(y_2^2,y_3^{3})+ \dots + \wtG(y_k^{k},x_{k+1})\\
            &\leq \alpha\wG(x_1,y_2^{2}) +\alpha\wG(y_2^2,y_3^{3})+ \dots + \alpha\wG(y_k^{k},x_{k+1})\\
    &= \alpha \wG(J^*) = \alpha \w(W^*).
  \end{align*}
  If the %
  described procedure
  fails to produce an $\alpha$\hyp approximate solution for~$I$,
  then, by contraposition,
  $\algA$~failed to produce an $\alpha$\hyp approximate
  solution for at least one \stRPP{} instance~$R[u,v,i]$
  corresponding to an arc~$(u_i,v_{i+1})$ on~$J^*$.
  Since $J^*$~has $k$~arcs,
  this happens with probability at most~$kp$
  by the union bound.
\end{proof}

\section{A 5/3-approximation algorithm for \hcpl}
\label{sec:apx}

\noindent
We now show a
polynomial\hyp time
5/3\hyp approximation algorithm for \stRPP{},
which, by \cref{thm:strppApprox},
carries over to \hcpl{}.
The algorithm is an adaption
of the \citeauthor{Chr76}\hyp \citeauthor{Ser78}\hyp like
3/2\hyp approximation algorithm
from RPP~\cite{EGL95,BNSW15}
to \stRPP{}.
It closely follows
\citeauthor{Hoo91}'s \citep{Hoo91} adaption of the
\citeauthor{Chr76}-\citeauthor{Ser78} 3/2\hyp approximation
algorithm 
from metric TSP \citep{Chr76,Ser78,BS20b} to metric \stTSP.

\begin{theorem}\label{thm:53stRPP}
  The \stRPP{} is $5/3$\hyp approximable in $O(n^3)$~time.
\end{theorem}
\begin{proof} %
  We assume $s\ne t$
  (otherwise,
  one can add a dummy vertex $s\ne t$
  and an edge~$\{s,t\}$
  of zero weight to the initial graph).
  We only show the 5/3\hyp approximation algorithm
  for \stRPP{} instances~$I:=(G, R, \w, s,t)$
  such that $G=(V,E)$~is a complete
  graph on the vertex set~$V=V(R)\cup\{s,t\}$ and
  such that the weight function~$\w$ satisfies
  the triangle inequality.
  This is enough,
  since the general
  case reduces to this special
  case in $O(n^3)$~time and any
  $\alpha$\hyp approximation for the special case
  yields an $\alpha$\hyp approximation
  for the general case \citep{BNSW15}.
  The 5/3\hyp approximation algorithm works in four steps.
  
  \begin{inparaenum}[\em Step 1.]
  \item 
    Compute a set~$T\subseteq E$ of edges
    of minimum total weight such that
    $R\cup T$ forms a spanning connected subgraph of~$G$ %
    (for example, using \citeauthor{Kru56}'s algorithm
    \citep{Kru56}).

  \item  %
    Let $S\subseteq V$~be the set of
    vertices in~$V\setminus\{s,t\}$
    that are imbalanced in~$\sg{G}{R\cup T}$
    and of those vertices in~$\{s,t\}$
    that are balanced in~$\sg{G}{R\cup T}$.
    Note that $|S|$~is even:
    Indeed, consider the set~$S'\subseteq V$
    of all vertices that are imbalanced in~$\sg{G}{R\cup T}$.
    Clearly, $|S'|$ is even.
    Now, if $s,t \in S'$, then $S=S'\setminus \{s,t\}$.
    If $s,t \notin S'$, then $S=S'\cup \{s,t\}$.
    If $s\in S'$ and $t \notin S'$ (or vice versa),
    then $S = S'\cup \{t\} \setminus  \{s\}$ (or $S = S'\cup \{s\} \setminus  \{t\}$).
    Thus, $|S|$~is even.

  \item  %
    Construct a minimum\hyp weight
    perfect matching~$M\subseteq E$
    on the vertices of~$S$ in~$G$
    (for example, using \citeauthor{Law76}'s
    algorithm \citep[Section~6.10]{Law76}).
    
  \item Return an Euler walk~$P$
    in $\sg{G}{R\uplus T\uplus M}$.
    Note that~$P$ exists
    (and can be computed using
    Hierholtzer's algorithm \citep{HW73,Fle91})
    since
    $\sg{G}{R\uplus T\uplus M}$~is
    connected and all its vertices
    except for~$s$ and~$t$
    are balanced.
    Thus,
    the endpoints of~$P$
    are~$s$ and~$t$
    and $P$~is a feasible solution to~$I$.
  \end{inparaenum}
  
  All steps can be carried out in $O(n^3)$~time.
  It remains to prove
  that $P$~is a 5/3\hyp approximation.
  To this end,
  let $P^*$~be an optimal solution for~$I$.
  Obviously, $\w(R\cup T) \le \w(P^*)$.
  Thus,
  it remains to show $\w(M) \le 2/3\cdot \w(P^*)$.  
  To this end,
  consider~$Q = E(P^*) \uplus R \uplus T$.
  We will construct three
  perfect matchings~$M_1$, $M_2$, and~$M_3$ on~$S$ in~$G$
  such that~$\w(M_1)+\w(M_2)+\w(M_3)\leq\w(Q)$,
  and thus $\w(M) \le 1/3\cdot \w(Q) \le 2/3\cdot \w(P^*)$.

  \looseness=-1
  Since the imbalanced vertices of~$\sg{G}{P^*}$
  are exactly~$s$ and~$t$,
  the imbalanced vertices in $\sg{G}{Q}$
  are exactly those in the set~$S$.
  Let the vertices of $S=\{v_1,v_2,\dots,v_{2\ell}\}$
  be numbered in the order of their first
  occurrence on~$P^*$
  and let $P_i^*$~be the subwalk of~$P^*$
  between the vertices~$v_{2i-1}\in S$ and~$v_{2i} \in S$
  for all $i\in\{1, \dots, \ell\}$.
  Let \[
    E_1:=\biguplus_{i=1}^\ell E(P_i^*).
  \]
  By shortcutting each path~$P_i^*$ to one edge,
  one gets a perfect matching~$M_1$
  on the vertices of~$S$
  such that $\w(M_1)\leq \w(E_1)$.

  The subgraph~$\sg{G}{Q \setminus E_1}$ is  Eulerian: %
  it is connected
  since $  R \uplus T \subseteq Q \setminus E_1$
  and it is balanced since
  the imbalanced vertices of~$\sg{G}{E_1}$
  are exactly those of~$\sg{G}{Q}$, that is, $S$.
  Its Euler cycle
  can be shortcut to a simple cycle on~$S$,
  which can be partitioned into two
  perfect matchings~$M_2$ and~$M_3$ on~$S$. 
  Thus,
  \begin{align*}
    \w(P) &= \w(R\cup T) + \w(M) \\
    &\le \w(P^*) + (\w(M_1) + \w(M_2) +\w(M_3) )/3 \\ %
    &\le \w(P^*) + (\w(E_1) + \w(Q\setminus E_1))/3 \\
    &\leq \w(P^*) + \w(Q)/3 \le  5/3 \cdot \w(P^*), %
  \end{align*}
  where the second inequality is due to the metric weights $\w$.\end{proof}
\noindent
Plugging \cref{thm:53stRPP}
into \cref{thm:strppApprox},
we immediately get:
\begin{corollary}
  \hcpl{} is $5/3$\hyp approximable in $O(kn^5)$~time.
\end{corollary}

\section{Parameterized algorithms for \hcpl}
\label{sec:fpt}

\noindent
\cref{thm:strppApprox} allows us to easily
transfer well\hyp known parameterized algorithms
from \RPP{} to \hcpl{} to show:

\begin{theorem}
  Let $\w_{\max}$~be the maximum edge weight
  and $c$~be the maximum number of connected components
  in any edge class of an \hcpl{} instance.
  Then, \hcpl{} is
  \begin{compactenum}[i)]
  \item polynomial\hyp time solvable for constant~$c$ and
    
  \item solvable in $2^c\cdot\poly(\w_{\max},n)$~time
    with exponentially decreasing error probability.
  \end{compactenum}
\end{theorem}

\begin{proof}
  To prove the theorem,
  it is enough to show that the known \RPP{} algorithms
  can also be used for~\stRPP{}.
  To this end,
  we reduce \stRPP{} to~\RPP{}.
  We assume that $s\ne t$
  and that~$s$ and~$t$ are non-adjacent
  in \stRPP{} instances
  (otherwise, we can add a new source~$s'$
  and a required weight\hyp zero edge~$\{s',s\}$).
  
  Now,
  note that an \stRPP{} instance~$I:=(G,R,\w,s,t)$
  can be reduced to an \RPP{} instance~$I':=(G',R',\w')$
  where $I'$~is obtained from~$I$
  by adding an edge~$\{s,t\}$ of weight $2\w(E)$
  to both~$E$ and~$R$.
  Then,
  an optimal solution~$P$ for~$I$
  yields a solution of weight~$\w(P)+2\w(E)$ for~$I'$.
  Moreover,
  an optimal solution~$P$ for~$I'$
  uses the edge~$\{s,t\}$ exactly once:
  if $P$~traversed it multiple times,
  then it would be cheaper
  to replace the second traversal of~$\{s,t\}$
  by any other $s$-$t$-path in~$G$.
  Thus,
  $P$~can be turned into a solution
  of weight~$\w(P)-2\w(E)$ for~$I$.
  That is,
  optimal solutions translate between~$I$ and~$I'$
  (yet approximate solutions do not).
  
  Moreover,
  if the number of connected components in~$\sg{G}{R}$
  is~$c'$,
  then the number of connected components in~$\sg{G'}{R'}$
  is at most~$c'+1$.
  Thus,
  since \RPP{} is solvable in polynomial time for constant~$c'$ \citep{Fre77,BNSW15}, so is \stRPP{}.
  And since \RPP{} is solvable in $2^{c'}\cdot\poly(\w_{\max},n)$ with
  exponentially decreasing
  error probability~\citep{GWY17},
  so is \stRPP{}.
  To conclude the proof of the theorem,
  it is enough to apply \cref{thm:strppApprox}
  and to observe
  that,
  for any \stRPP{} instance~$P[u,v,i]$
  solved,
  the subgraph~$\sg{G}{E_i}$ induced
  by the required edges~$E_i$
  has at most $c$~connected components.
\end{proof}

\section{Conclusion}
\noindent
Our work leaves open several questions.
First,
what is the computational complexity of \hcpc{}
with a constant number of edge classes?
It has been conjectured to be polynomial\hyp time
solvable \citep{BNSW15},
yet no polynomial\hyp time algorithm
is known even for the case with three classes.

Second,
can one close the gap between our 5/3\hyp approximation for \stRPP{}
and the known 3/2\hyp approximation for \RPP{}~\cite{EGL95,BNSW15}?
For example,
recently,
a 3/2\hyp approximation for metric \stTSP{}
has been shown \citep{Zen19},
matching the approximation factor
of the
Christofides\hyp Serdyukov algorithm
for metric TSP \citep{Chr76,Ser78,BS20b}.
It is not obvious whether the used approaches
carry over to \stRPP{},
yet closing the gap
between \stRPP{} and \RPP{}
would immediately  give a 3/2\hyp approximation for \hcpl{}.

Our fixed\hyp parameter algorithm
for \hcpl{} parameterized
by the maximum number~$c$ of connected components
in any edge class
raises the question
whether and how lossy kernelization results
for RPP parameterized by~$c$ \citep{BFT20b} 
carry over to \hcpl{}.

\paragraph{Funding}
This work is funded by  Mathematical Center in Aka\-demgorodok,
agreement No.\ 075-15-2019-1675 with the Ministry of Science and
Higher Education of the Russian Federation.

\bibliographystyle{hs-lintimespace}
\bibliography{hcpp}

\end{document}